\documentclass[letterpaper,twocolumn,fleqn]{article} 
\usepackage{ist}
\usepackage{graphicx} 
\usepackage{biblatex}
\usepackage[title]{appendix}
\usepackage[hidelinks]{hyperref}
\usepackage[usestackEOL]{stackengine}
\usepackage{outlines}  
\usepackage{bm} 
\usepackage{amssymb}  
\usepackage{mathtools}  
\usepackage{tikz}  
\usepackage{amsthm}
\usepackage{float} 
\usepackage{caption} 
\usepackage{subcaption}  
\usepackage{booktabs}  

\newcommand{\x}{{\bm{x}}}  
\newcommand{\y}{{\bm{y}}}  
\newcommand{\z}{{\bm{z}}}  
\newcommand{\btheta}{{\bm{\theta}}} 
\newcommand{\f}{{\bm{f}}} 
\renewcommand{\S}{{S}}  
\renewcommand{\H}{{\bm{H}}}  
\newcommand{\T}{T_{>0}}  

\newcommand{\xmax}{{\bm{x}_\text{max}}}  

\DeclareMathOperator*{\argmin}{arg\,min}
\DeclareMathOperator*{\argmax}{arg\,max}

\newtheorem{theorem}{Theorem}

\theoremstyle{definition}
\newtheorem{definition}{Definition}

\title{Supervised Reconstruction for Silhouette Tomography}

\author{ 
Evan Bell$^a$, Michael T. McCann$^b$, Marc Klasky$^b$\\
$^a$Michigan State University; East Lansing, MI\\
$^b$Theoretical Division, Los Alamos National Laboratory; Los Alamos, NM
}

\begin{document}

\maketitle

\begin{abstract}
    In this paper, we introduce \textbf{silhouette tomography},
    a novel formulation of X-ray computed tomography that relies only on the geometry of the imaging system.
    We formulate silhouette tomography mathematically and provide a simple method for obtaining a particular solution to the problem, assuming that any solution exists.
    We then propose a supervised reconstruction approach that uses a deep neural network to solve the silhouette tomography problem. We present experimental results on a synthetic dataset that demonstrate the effectiveness of the proposed method.
\end{abstract}

\section{Introduction}
X-ray computed tomography (CT)
is an imaging modality that involves collecting X-ray projection images
of the same object from many angles
and using them to reconstruct a 3D image of the object.
It has applications in medicine~\cite{webb_introduction_2002},
nondestructive testing~\cite{maire_quantitative_2014},
and security~\cite{wu_object_2023}.
When the number of projection angles is large
and the scanner is well characterized in terms of
its geometry,
the spectrum of the X-rays used,
and the spectral response of the detector,
CT reconstruction may be accurately performed by either linear~\cite{bracewell_inversion_1967} methods
or model-based iterative reconstruction~\cite{yu_fast_2011,mcgaffin_alternating_2015}.

More recently,
research on CT reconstruction has focused on so-called \emph{extreme imaging},
wherein the projections are limited in number~\cite{niu_sparse_2014,lee_view_2017},
angular coverage~\cite{wei_joint_2018},
or signal-to-noise ratio~\cite{li_super_2019,shan_competitive_2019}.
In these settings, heavy use of prior information---%
either in the form of regularization or machine learning%
---is needed to form an adequate reconstruction.
Especially in the case of learning-based approaches to CT reconstruction,
many methods  have been developed and evaluated in a setting where 
both testing and training data are generated in silico
and use exactly the same imaging model (sometimes called \emph{inverse crime}).
While this approach is useful at the proof-of-concept stage,
it does not provide an accurate picture of how these methods would perform
on real, experimental data.

In this work,
we present \emph{silhouette tomography (ST)},  a novel formulation of the X-ray CT reconstruction problem
that sidesteps the need to accurately model a CT system in order to
perform model-based reconstruction or generate training data for learning-based reconstruction.
In ST,
rather that relying on the grey-level values of the measured projections,
we instead only use them to determine if an object of interest is present
in each pixel (equivalently, along each ray).
In doing so,
we transform the CT reconstruction problem into one that is purely geometric,
and therefore only requires knowledge of the geometry of the imaging system
(e.g., the projection angles and the cone angle).
This simplifies system calibration
and increases the chance that a learning-based algorithm trained on simulation data
will generalize well to real, experimental data.
This approach has the added benefit of enabling reconstruction of objects in cluttered scenes
and/or in the presence of highly-absorbing objects,
provided that the object of interest can be accurately segmented in each projection view.
On the other hand,
choosing the ST formulation instead of the CT one
generally leads to a less well-posed and therefore more challenging reconstruction problem.

\textbf{Related work.}
In binary tomography~\cite{weber_binary_2004,weber_benchmark_2006,Varga2011,kadu_convex_2020},
the object to be reconstructed is binary,
but---distinct from ST---%
the measurements are real-valued.
In shadow tomography~\cite{savarese_shape_2005},   
the  the actual shadows cast by an object,
rather than its X-ray projections,
are used to determine its shape.

\textbf{Outline.}  
In the next section,
we formulate the silhouette tomography reconstruction problem
and define the maximal reconstruction,
which provides a simple way of finding at least one solution to any ST problem.
We then present a learning based-approach to ST reconstruction.
Finally,
we provide an experimental comparison of these methods
and conclude.

\section{Silhouette Tomography}
The silhouette tomography problem 
is to recover the nonzero support of an object
from its binarized X-ray projections.
Mathematically,
we represent the object as a binary vector $\x \in \{0, 1\}^{N}$,
with either $N = N_1 \times N_2$ (height $\times$ width, 2D problems)
or $N = N_1 \times N_2 \times N_3$ (depth $\times$ height $\times$ width, 3D problems).
Let $\x[n] = 1$ if the object occupies pixel/voxel $n$ 
and $\x[n] = 0$ otherwise.
The binarized X-ray projection measurements are
$\y \in \{0, 1\}^M$,
with $M = V \times M_1$ (number of views $\times$ length, 2D problems)
or $M = V \times M_1 \times M_2$ (number of views $\times$ height $\times$ width, 3D problems).
Let $\y[m]=1$ if the object occupies any of the pixels/voxels
along ray $m$
and $\y[m]=1$ otherwise.
We can relate the projections $\y$ to the object $\x$
via the silhouette tomography operator,
$\S : \mathbb{R}^N \to \{0, 1\}^M$,
defined as
\begin{equation} \label{eq:forward}
    \y = \S(\x) = \T(\H \x),
\end{equation}
where
$\H \in \mathbb{R}^{M \times N}$
is a linear X-ray transform operator
with $\H[m, n] \ge 0$
and $T_{>0}$ is a thresholding operator defined by
\begin{equation}
     \T(\x)[n] = \begin{cases}
         1 & \x[n] > 0; \\
         0 & \text{otherwise}.
     \end{cases}
\end{equation}

In contrast to the linear tomography problem,
the silhouette tomography problem generally admits multiple solutions even when the number of views is large.
The linear tomography problem is to recover a real-valued object $\x \in \mathbb{R}^{N}$
from its real-valued measurements $\y \in \mathbb{R}^M$
given by 
\begin{equation}
    \label{eq:linear}
    \y = \H \x,
\end{equation}
where $\H \in \mathbb{R}^{M \times N}$ is again
a linear X-ray transform operator.
For this linear problem,
the solution is unique when $M \ge N$ because $\H$ can be assumed to have linearly independent rows.
In silhouette tomography,
we have strictly less information because the measurements are binary.
While a theoretical analysis of its ill-posedness is beyond the scope of this paper,
it is easy to come up with examples of shapes that cannot be distinguished
by measurements of the form \eqref{eq:forward},
see, e.g.,  Figure~\ref{fig:multiple_solutions}.

\begin{figure}
    \centering
    \hfill{}
\begin{tikzpicture}
\draw [gray, fill] plot [smooth cycle] coordinates {(0,0) (1,1) (3,1) (3,0) (2,-1)};
\draw [black, thick] (-.5, -.5) -- (1.5, 1.5);
\end{tikzpicture}\hfill{}
\begin{tikzpicture}
\draw [gray, fill] plot [smooth cycle] coordinates {(0,0) (1,1) (3,1) (3,0) (2,-1)};
\draw [white, fill] plot [smooth cycle] coordinates {(1.4,1) (2,.5) (2.5,1) (2,2)};
\draw [black, thick] (-.5, -.5) -- (1.5, 1.5);
\end{tikzpicture}
\hfill{}      
    \caption{Silhouette tomography generally has multiple solutions.
    For example, each ray intersecting the left object also intersects the right one
    and each ray not intersecting the left object also does not intersect the right one.
    }
    \label{fig:multiple_solutions}
\end{figure}
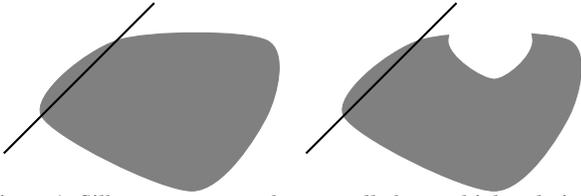

For any silhouette tomography problem,
we now show that
there is always one solution that is uniquely defined and easy to compute;
we call it the maximal solution.
\begin{definition}[Maximal solution]
    The \emph{maximal solution}, $\xmax(\y)$,
    of \eqref{eq:forward} 
is the $\x$ that has the maximal norm among
all $\x$'s that satisfy \eqref{eq:forward},
\begin{equation}
    \xmax(\y) = \argmax_\x \|\x \|_2^2 \quad \text{s.t.} \quad \S(\x) = \y.
\end{equation}
We will sometimes drop the argument $(\y)$
when it is clear from context.
\end{definition}

\begin{theorem}[Computing the maximal solution]
When any solution to \eqref{eq:forward} exists,
the maximal solution of \eqref{eq:forward} is given by
\begin{equation} \label{eq:xmax}
    \xmax(\y) = \neg \T(\H^T (\neg \y)),
\end{equation}
where $\neg$ denotes Boolean negation (i.e., swapping 0 and 1).    
\end{theorem}
\begin{proof}
We first show that $\xmax$ satisfies \eqref{eq:forward}.
Equation \eqref{eq:forward} implies 
that $\y[m]=1$ if and only if
that there exists an $n$
such that $\H[m,n] > 0$ and $\x[n]=1$.
Conversely,
$\y[m]=0$ if and only if
for all $n$ such that $\H[m,n] > 0$, $\x[n]=0$.
The maximal reconstruction has $\xmax[n] = 1$
for all $n$ except those that fall into the later case,
i.e., it is only zero where it must be:
by the construction \eqref{eq:xmax},
$\xmax(\y)[n] = 0$ if and only if $\H^T (\neg \y)[n] > 0$,
which implies that there exists $m$ such that $\y[m]=0$ and $\H[m, n] > 0$.
Thus, if a solution to \eqref{eq:forward} exists,
$\xmax$ is a solution to \eqref{eq:forward}.

Next, we show that $\xmax$ has maximal norm among all solutions to \eqref{eq:forward}.
Suppose $\z \in \{0,1\}^N$ with $||\z||_2^2 > ||\xmax||_2^2.$ 
Then, there exists some $n'$ such that $\z[n']=1$ but $\xmax[n']=0$. Then, by the same reasoning as above, there exists an $m'$ such that $\y[m'] = 0$ and $\H[m',n'] > 0$. We can then calculate
$(\H\z)[m'] = \sum_{j=1}^{N} \H[m',j]\z[j] \geq \H[m',n']\z[n'] > 0.$ Thus $\T(\H\z)[m'] = 1$, but $\y[m'] = 0$. Hence $\z$ is not a solution to \eqref{eq:forward}. This shows that $\xmax$ has the largest norm of any solution to \eqref{eq:forward}, as claimed.
\end{proof}

\textbf{Preparing Data.}
We now describe how ordinary X-ray CT projection data
can be processed so that it can be reconstructed using ST.
In the simplest case
wherein we are imaging a single object,
it suffices to simply threshold the projection images,
creating binary masks indicating the presence or absence of the object at each pixel in each projection.
If the goal is to reconstruct multiple objects,
then each one should be segmented separately
and one ST reconstruction performed for each.
Automating this segmentation process is outside the scope of the current work.

\section{Proposed Method: Supervised Reconstruction}\label{sec:proposed}
We propose a deep learning-based approach to solving the ST reconstruction problem. The approach is based on using a database of paired objects, $\x$,
and measurements, $\y$,
along with a known linear X-ray transform operator, $\H$,
to learn to reconstruct $\x$ from $\y$.
Thus, given $D$ tuples of training data $(\x_d, \y_d, \H), \ d=1,...,D$, 
we train our network weights according to
\begin{equation}
    \label{eq:training_objective}
    \btheta^{\bm{*}} = \argmin_\btheta \frac{1}{D} \sum_{d=1}^D ||\x_d - \f_\btheta(\H^T \y_d)||_2^2,
\end{equation}
where $\f$ is a neural network with parameters $\btheta$.

\textbf{Proposed Architecture.} Because the ST reconstruction problem is highly ill-posed, the choice of $\f$ is important for obtaining a useful solution to $\eqref{eq:training_objective}$. The model we propose is a deep U-Net based on the architecture described in \cite{deblurring_unet}.

\begin{figure*}
    \centering
    \includegraphics[trim={1in 1in 1in 1in},clip,width=\linewidth]{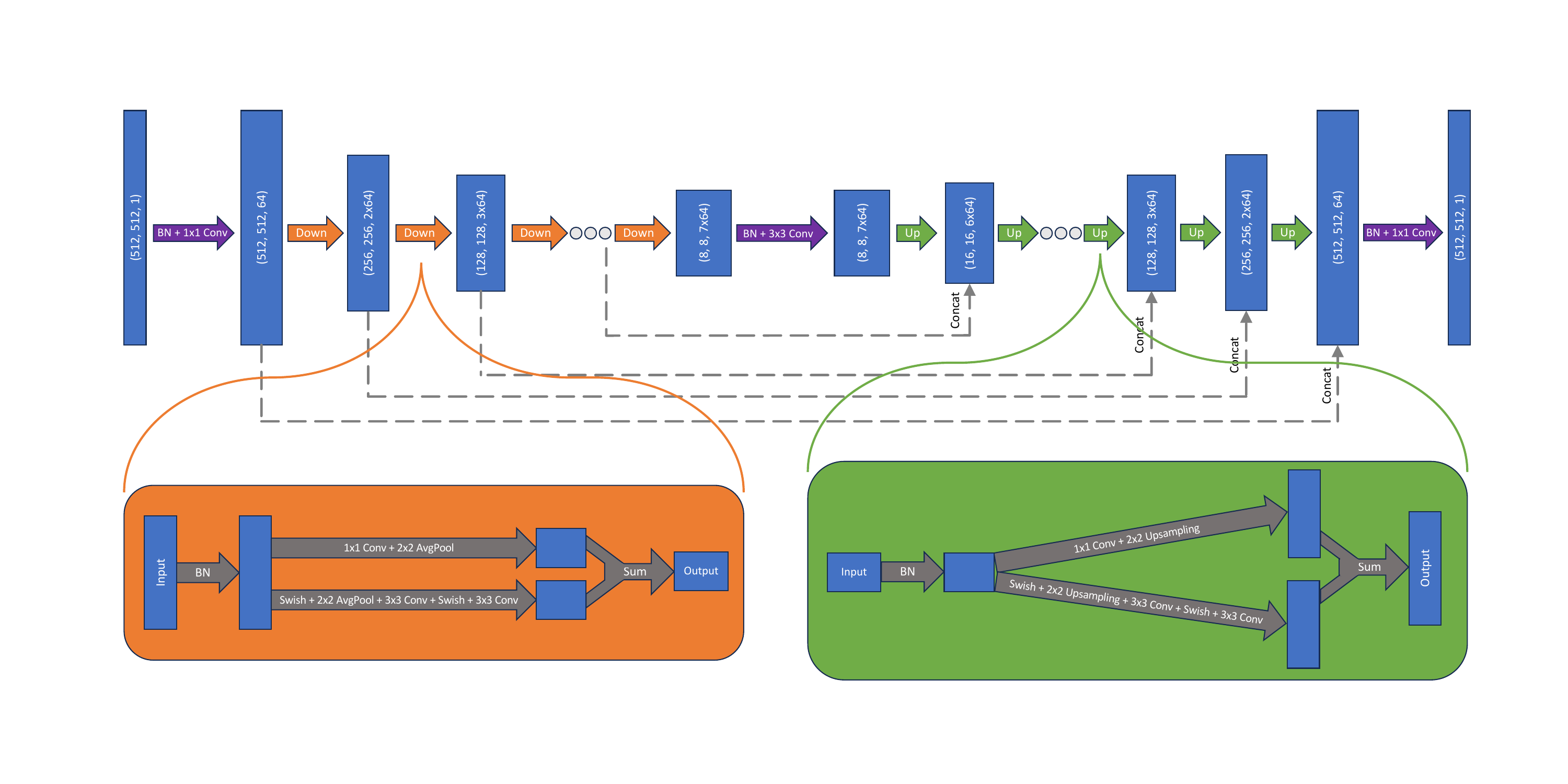}
    \caption{Schematic of the proposed neural network architecture. The internal structure of each downsampling and upsampling block is shown in the orange and green boxes respectively.}
    \label{fig:unet}
\end{figure*}

A complete schematic of the proposed architecture is shown in Figure \ref{fig:unet}. The proposed architecture has six down-sampling blocks and six up-sampling blocks.
These blocks are inspired by the residual blocks used in BigGAN \cite{biggan}.
The internal structure of each of these blocks is shown in the figure. 
In each ``down" block, the number of channels increases by 64, and the spatial dimension of the features decreases by a factor of two.
In the ``down" blocks, the first $3\times3$ convolution increases the number of channels. 
In each ``up" block, the number of output channels decreases by 64 relative to the previous block, and the spatial dimension of the features increases by a factor of two.
In the ``up" blocks, the first $3\times3$ convolution decreases the number of channels.
The network also includes skip connections, indicated by the gray dashed arrows in the figure.
Features extracted in the down-sampling path are concatenated with those of the same size in the up-sampling path before each ``up" block. 

Batch normalization~\cite{ioffe_batch_2015} is used throughout the network. This operation is indicated by ``BN" in the figure. The sigmoid linear unit, also known as ``SiLU" or ``Swish," is used as the activation function throughout the network, and is defined pointwise for $x \in \mathbb{R}$ by
\begin{equation}
    \text{Swish}(x) = x \cdot \sigma(x) = \frac{x}{ 1 + e^{-x}},
\end{equation}
where $\sigma$ is the standard logistic sigmoid function.

In total, our proposed architecture has approximately $22.3$ million trainable parameters.

\section{Experiments and Results}\label{sec:experiments}

\textbf{Dataset. } 
We created a synthetic dataset to test the proposed approach to silhouette tomography.
The dataset consisted of 2,036 3D objects
(e.g., airplanes, birdhouses, bottles, ...)
chosen at random
from the ShapeNet dataset \cite{shapenet2015}.
In order to generate binary ground truth volumes
that were suitable for our experiments,
we preprocessed each object in the following way:
We made the model watertight using code from \cite{huang_robust_2018}.
Then, we randomly rotated the model using Trimesh~\cite{dawsonhaggerty_trimesh_2019}
and voxelized on a $511\times511\times511$ grid
using the method of \cite{mescheder_occupancy_2019}%
\footnote{Available at \url{https://github.com/nikwl/inside_mesh}.}.

The dataset was then randomly split into 1,832 volumes for training, 102 volumes for validation, and 102 volumes for testing.

For each volume, we simulated X-ray measurements using parallel beam geometry with 8 equispaced views at angles in the range $[0, \pi)$. All X-ray projections and backprojections were simulated using the ASTRA toolbox~\cite{astra1}. 
With this choice of geometry, the reconstruction problem can be completely separated into 2D slices.

Most of the volumes in the dataset contained some slices with no object content.
These empty slices were not were not used during training, validation, or testing. Once these slices were removed, there were 698\,485 slices for training, 39\,628 slices for validation, and 38\,824 slices for testing.
An object from the testing set and one corresponding slice is shown in Figure \ref{fig:slice}.

\begin{figure}[H]
    \centering
    \includegraphics[width=\linewidth]{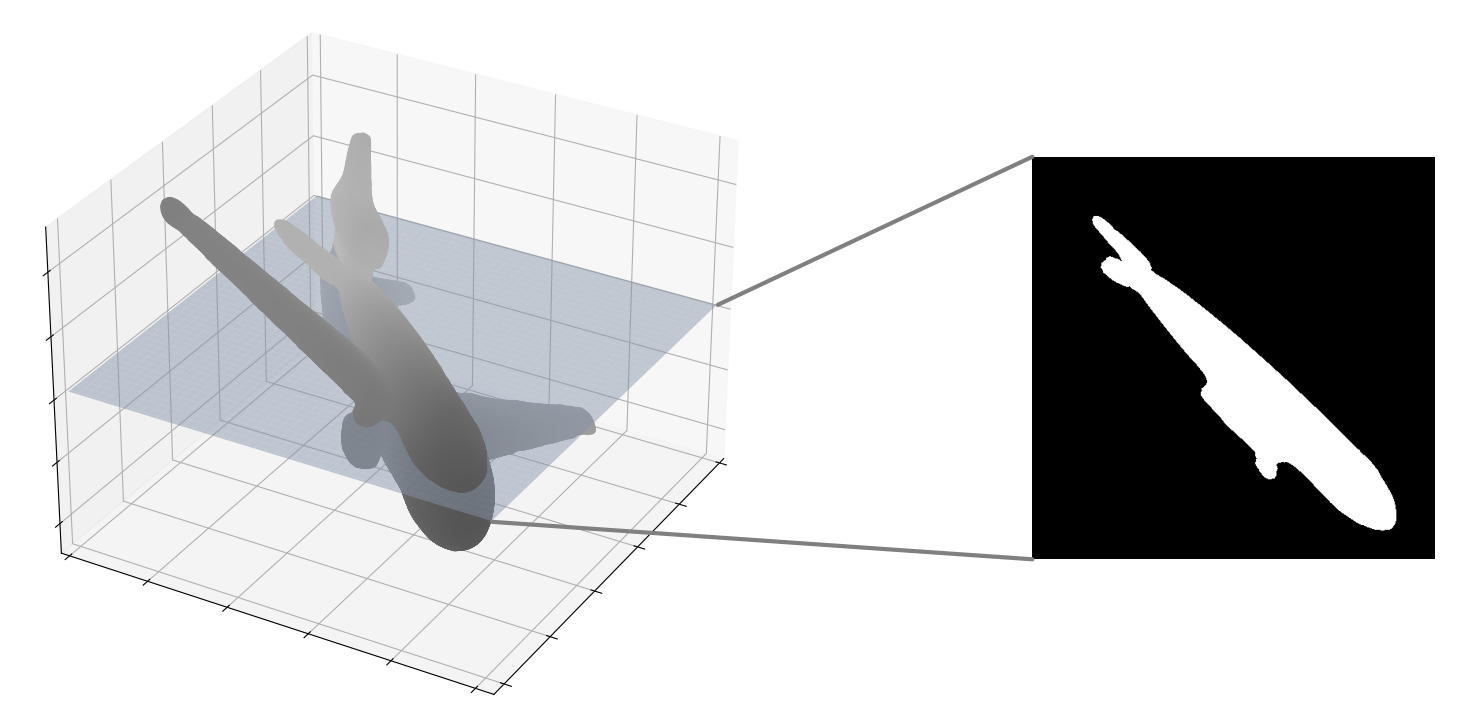}
    \caption{An object from the testing set and one slice from it.}
    \label{fig:slice}
\end{figure}

\textbf{Training details. } Because the proposed architecture repeatedly down-samples its input by a factor of two, it is most naturally suited to an input size with many factors of two. To accommodate this, every input to the network was zero-padded to a size of $512\times512$, and the output was cropped back to a size of $511\times511$.

The training objective was given by \eqref{eq:training_objective}.
The model was implemented using PyTorch and trained using the Adam optimizer with a learning rate of $10^{-5}$ and a batch size of 8.
The model was trained for a total of four passes over the training set (epochs),
after which both the training and validation losses had empirically converged.
The training took approximately 19 hours using eight NVIDIA GeForce RTX 2080 Ti GPUs.

\textbf{Quantitative comparison. }
Table \ref{tab:quant_results} gives a quantitative comparison of the reconstruction methods on the test set using three metrics: mean squared error (MSE), peak signal to noise ratio (PSNR), and the structural similarity index measure (SSIM). The numbers given in the table are the average for the three metrics across every slice of the test set. The SSIM was computed using the implementation in the TorchMetrics library with its default parameters \cite{metrics}. 
The row ``U-Net" gives metrics computed using the raw (floating point) output of the trained network.
The row ``binarized U-Net" gives metrics computed after thresholding the network's output at a value of $0.5$. There was one slice in the test set for which the thresholded neural network reconstruction had an MSE of 0, resulting in an infinite PSNR. The PSNR of this slice was discarded when computing the average PSNR over the test set.

\begin{table}
    \centering
    \begin{tabular}{llll}
    \toprule
        Method  & MSE  & PSNR  & SSIM \\ \midrule
        maximal  & 0.0628 & 15.29 & 0.921 \\
        U-Net & \textbf{0.0229} & \textbf{22.96} & 0.875 \\
        binarized U-Net & 0.0283 & 22.02 & \textbf{0.952} \\
        \bottomrule
    \end{tabular}
    \captionof{table}{Quantitative comparison of maximal reconstruction and the proposed method on a silhouette tomography reconstruction problem. The best result in each row is indicated in \textbf{bold}.}
    \label{tab:quant_results}
\end{table}


The neural network reconstruction significantly outperforms the maximal reconstruction in terms of both MSE and PSNR, both with and without the binarization. 
The maximal reconstruction outperforms the raw network output in terms of SSIM, but still falls short of the thresholded network output. This is likely because SSIM penalizes differences in overall contrast between images, where the variance of an image is used as a proxy for its contrast \cite{ssim}. Therefore the network is penalized for any uncertainty in its reconstructions, i.e. outputting pixel values between 0 and 1. Hence thresholding improves the SSIM of the supervised reconstruction, despite lowering its MSE and PSNR.

\textbf{Qualitative comparison. } Figure \ref{fig:qual_comp} gives a qualitative comparison of the maximal reconstruction and the learned reconstruction on a representative slice of the test set. It is evident that the maximal reconstruction significantly overestimates the object content of the slice, whereas the proposed method recovers some of the smooth structures present in the ground truth.


    
    

\begin{figure}
    \centering
    \begin{subfigure}[t]{0.48\linewidth}
        \includegraphics[width=\linewidth]{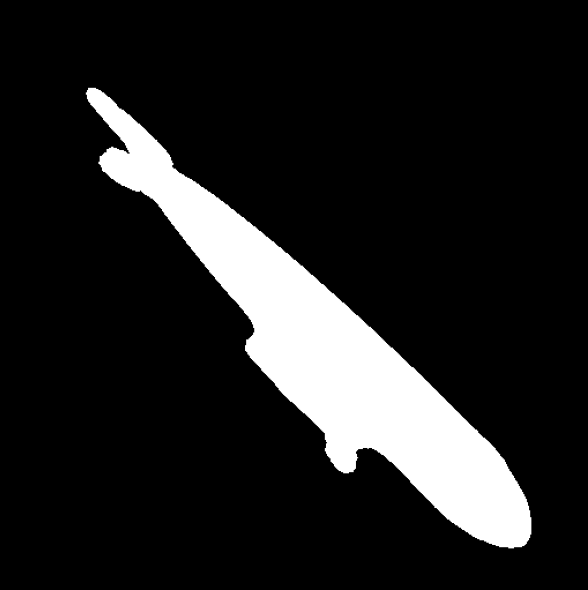}%
        \caption{\centering ground truth $\x$}%
    \end{subfigure}\vspace{2em}\\
    \begin{subfigure}[t]{0.48\linewidth}
        \includegraphics[width=\linewidth]{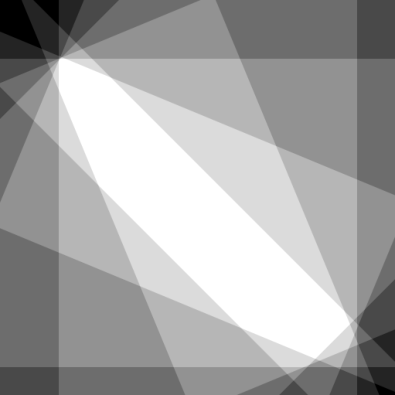}
        \caption{\centering network input $\H^T\y$}
    \end{subfigure}\hfill
    \begin{subfigure}[t]{0.48\linewidth}
        \includegraphics[width=\linewidth]{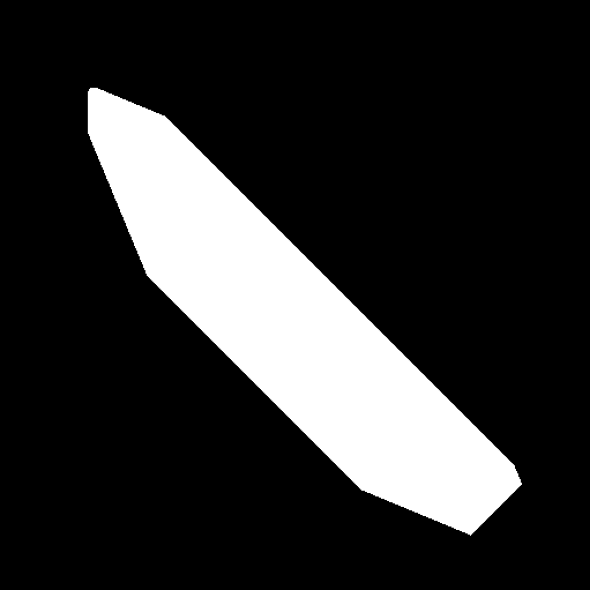}
        \caption{\centering maximal reconstruction $\x_\text{max}(\y)$}
    \end{subfigure}\vspace{2em}\\
    \begin{subfigure}[t]{0.48\linewidth}
        \includegraphics[width=\linewidth]{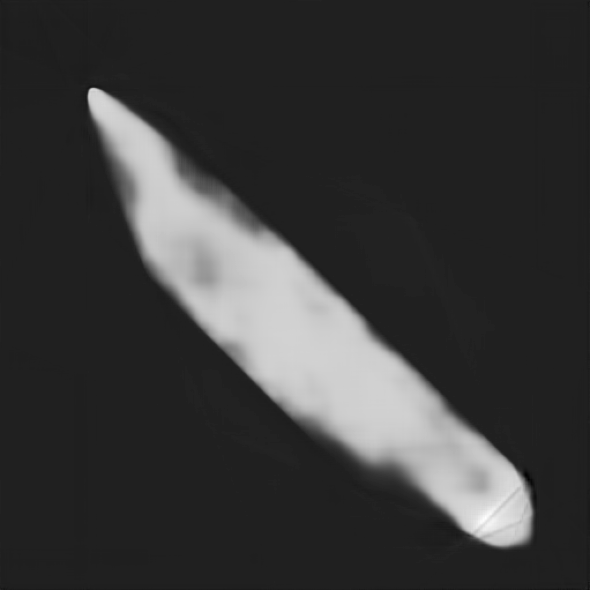}
        \caption{\centering NN reconstruction}
    \end{subfigure}\hfill
    \begin{subfigure}[t]{0.48\linewidth}
        \includegraphics[width=\linewidth]{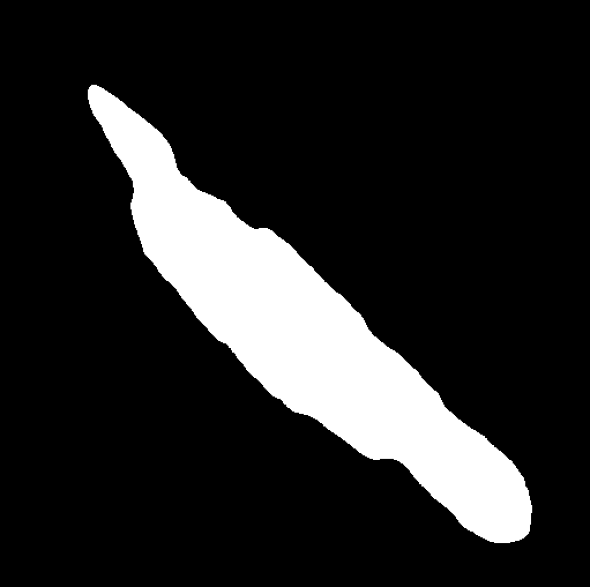}
        \caption{\centering NN reconstruction (binarized)}
    \end{subfigure}\vspace{1em}
    \caption{Qualitative comparison of the proposed reconstruction methods on one slice of the test set.}
    \label{fig:qual_comp}
\end{figure}

By reconstructing each slice of a 3D volume separately
and stacking them, we can perform 3D reconstruction.
Figure \ref{fig:3d_recon} gives a comparison of a full 3D volume of the test set reconstructed using both the maximal and proposed methods.

\begin{figure*}
    \centering
    
    \begin{subfigure}[t]{.3\linewidth}
        \includegraphics[width=\linewidth]{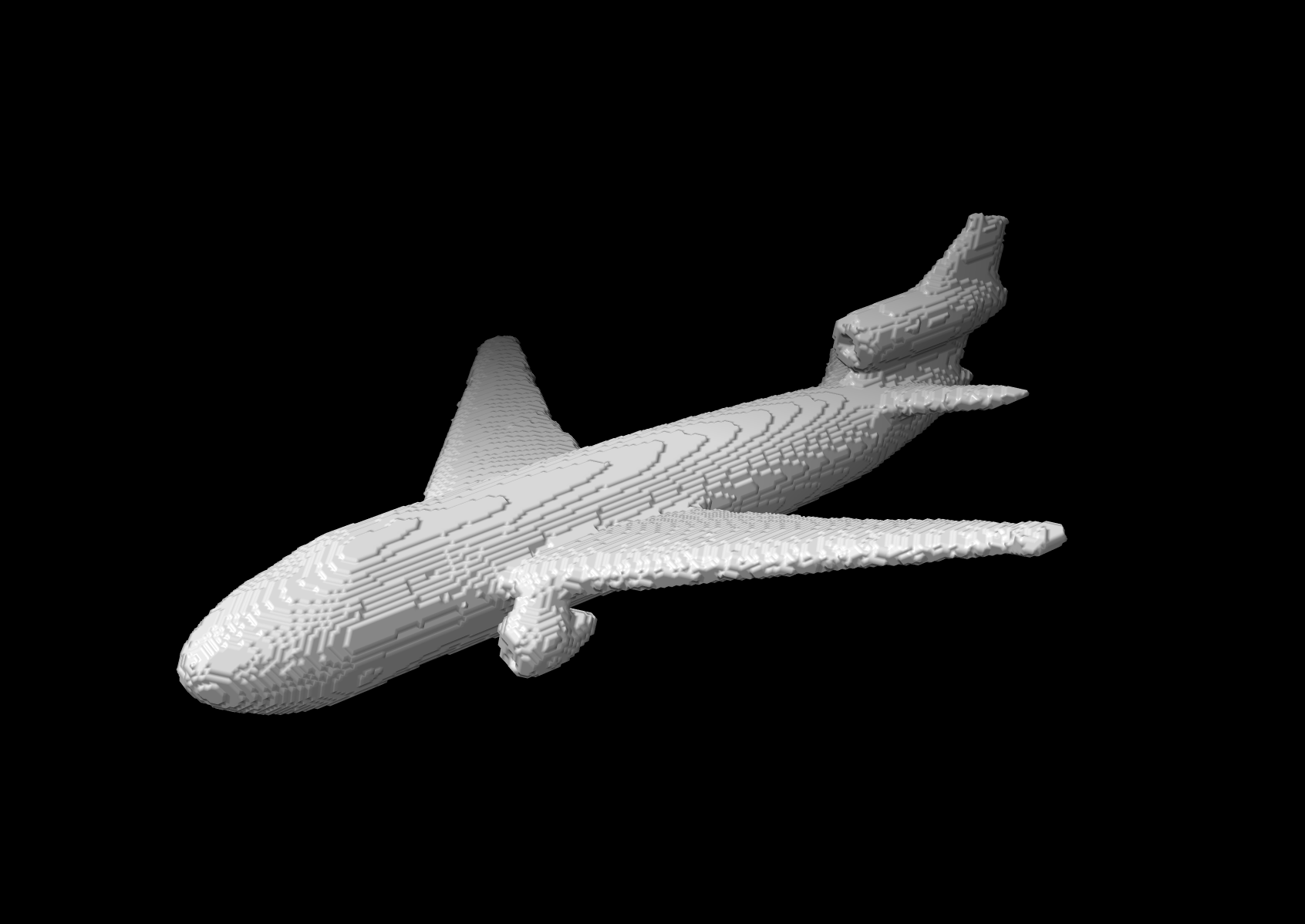}
        \caption{\centering ground truth}
    \end{subfigure}\hfill
    \begin{subfigure}[t]{.3\linewidth}
        \includegraphics[width=\linewidth]{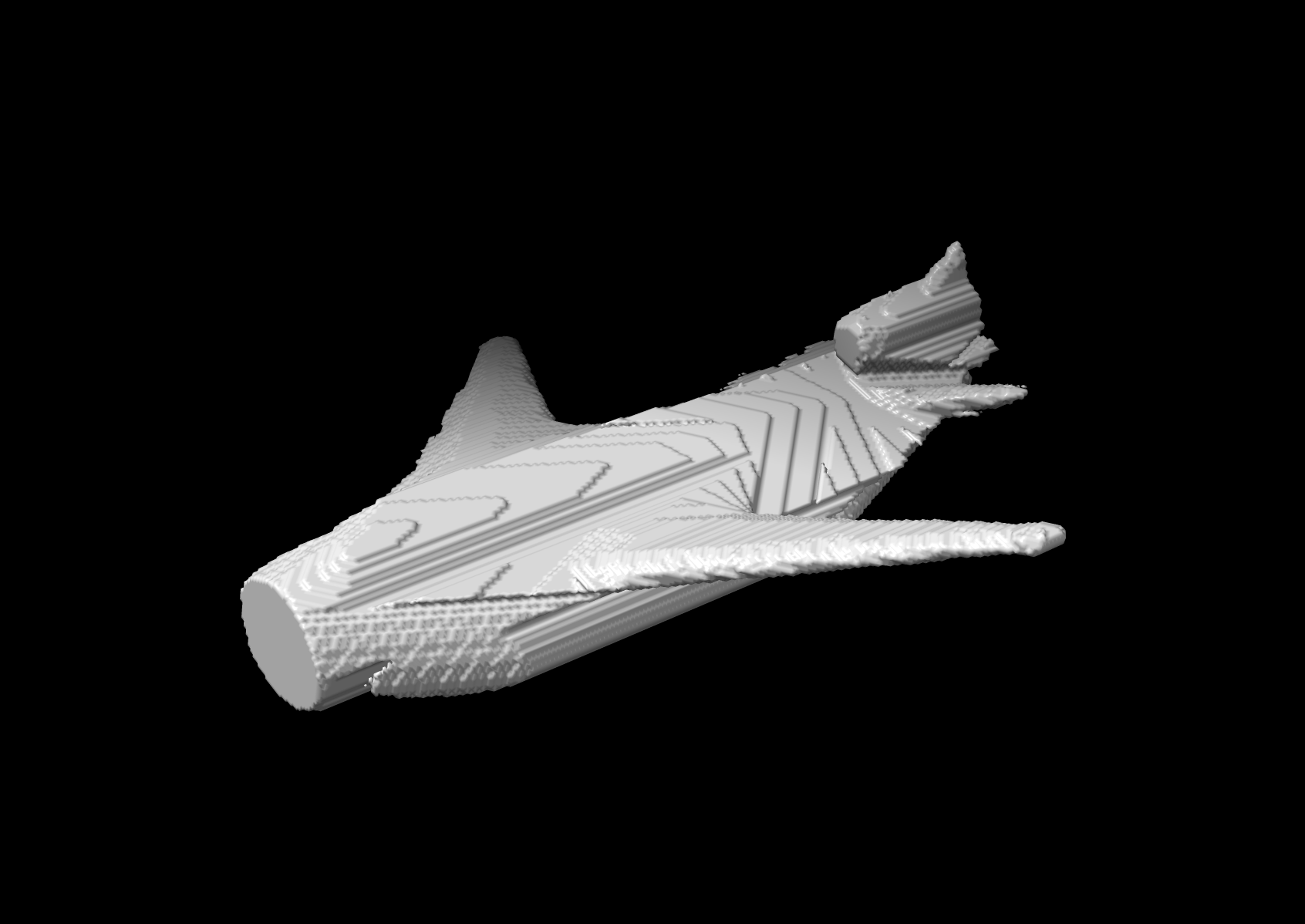}
        \caption{\centering maximal reconstruction}
    \end{subfigure}\hfill
    \begin{subfigure}[t]{.3\linewidth}
        \includegraphics[width=\linewidth]{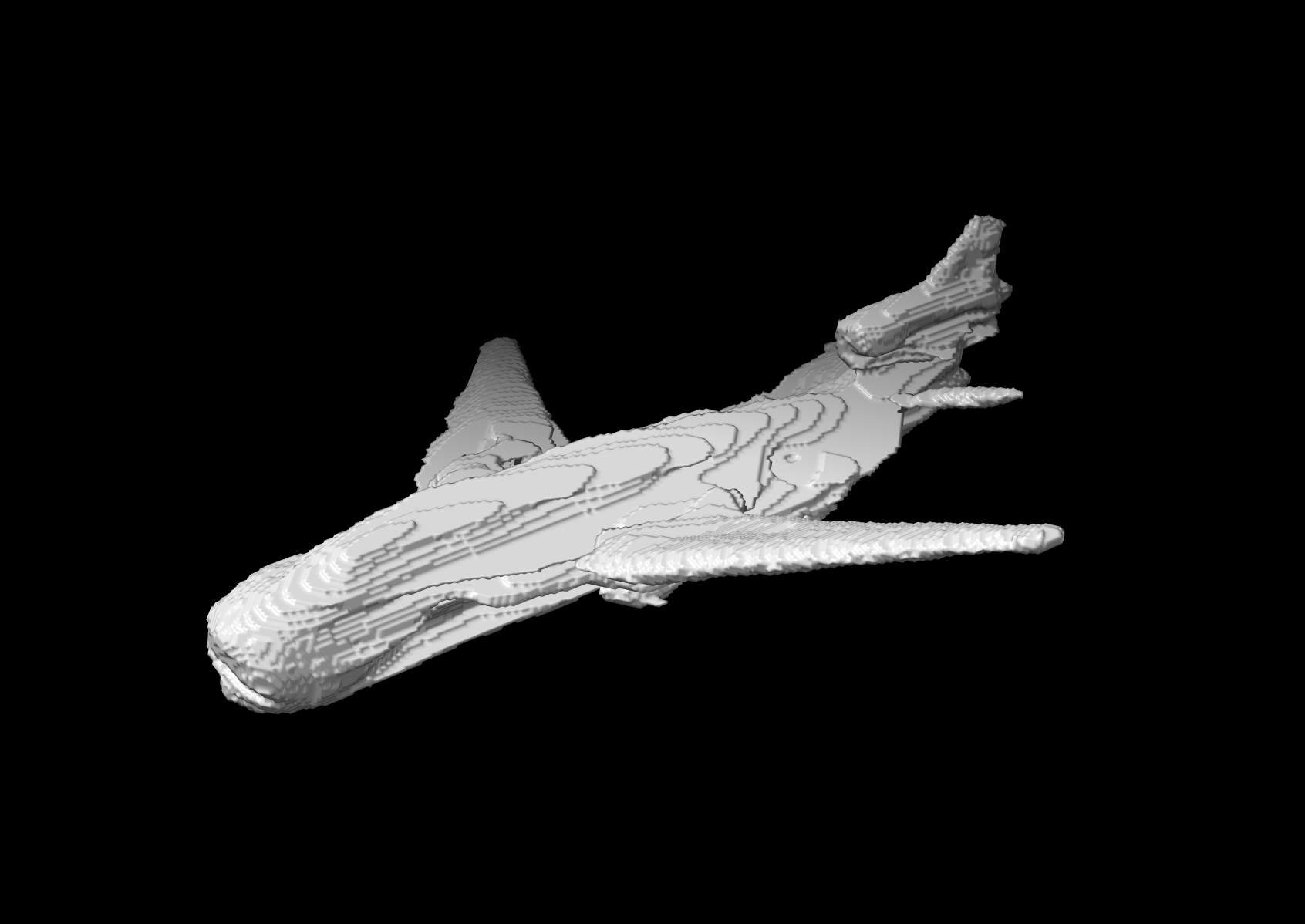}
        \caption{\centering NN reconstruction}
    \end{subfigure}
    \vspace{2em}
    \caption{Visual comparison of maximal and supervised approaches for 3D object reconstruction.}
    
    \label{fig:3d_recon}
\end{figure*}

\textbf{Comparison with linear tomography. } For extra context,
we include a comparison between the proposed supervised approach to silhouette tomography and a similar approach to linear tomography as formulated in \eqref{eq:linear}. To do this, the proposed architecture was trained with training objective \eqref{eq:training_objective} with $\y_d = \H\x_d$. All training settings, including optimizer, batch size, learning rate, and length of training were identical to the model trained for silhouette tomography.

Quantitative results are given in Table \ref{tab:linear_quant_results}. We again present metrics computed using both the floating point network output and the network output after thresholding at a value of 0.5. In this case, there were 74 slices in the test set for which the thresholded network output had an MSE of 0, giving an infinite PSNR. The PSNR of these slices was discarded when computing average PSNR on the test set.

The supervised model trained for linear tomography performs substantially better than the model trained for silhouette tomography across all metrics. This is unsurprising, because the ST problem is more ill-posed than the linear tomography problem. Moreover, we propose ST as a possible approach to tomography in settings where linear tomography is impractical. This experimental setup is ideal for linear tomography to succeed: the forward operator $\H$ is known perfectly, the measurements $\y$ are free of any additional noise, and the objects being imaged are completely uniform in material composition.

This comparison is valuable because it quantifies the size of the performance gap between linear and silhouette tomography under ideal conditions. Additionally, it demonstrates that the proposed neural network architecture is highly successful in the linear setting, which suggests that it is also a reasonable choice for silhouette tomography reconstruction.

\begin{table}
    \centering
    \begin{tabular}{llll}
    \toprule
        Method  & MSE  & PSNR  & SSIM \\ \midrule
        U-Net & \textbf{0.00123} & \textbf{34.20} & 0.982 \\
        binarized U-Net  & 0.00160 & 33.27 & \textbf{0.990} \\
        \bottomrule
    \end{tabular}
    \captionof{table}{Quantitative results for supervised linear tomography reconstruction using the proposed network architecture. The best result in each row is indicated in \textbf{bold}.}
    \label{tab:linear_quant_results}
\end{table}


\section{Discussion and Conclusion}\label{sec:discussion}

We introduced silhouette tomography (ST), a novel approach to X-ray CT that formulates the image reconstruction problem purely geometrically.
Whereas both traditional model-based reconstruction
and existing learning-based approaches to X-ray CT
rely on an accurate model of the imaging system---%
either for reconstruction or to generate realistic training data---%
the ST approach requires only geometric information,
i.e., the projection angles.
Rather than reconstructing a 3D absorption map,
ST reconstructs a binary map representing whether the object does or does not exist at each point in space.
This type of reconstruction is particularly well-suited to nondestructive testing and security applications
where object shapes are more important than density variations.

We mathematically formulated the ST problem and showed that when the problem has a solution, it always admits a unique solution that is maximal in norm.
We then proposed a supervised deep learning approach for ST reconstruction, and demonstrated that this approach significantly outperforms the maximal solution across multiple reconstruction metrics.
These results demonstrate that
supervised learning is an effective method for ameliorating
the ill-posedness of the ST reconstruction problem,
at least on our shape dataset.

Extensions to this work could involve
adapting the network architecture to handle fully 3D reconstruction,
performing experiments on real data,
and developing automated methods to segment objects in X-ray projection images.

\printbibliography

\end{document}